\documentclass[journal,12pt,onecolumn,draftclsnofoot]{IEEEtran}

\usepackage{amsmath}
\usepackage{amsthm}
\usepackage{amssymb}
\usepackage{graphicx}
\usepackage{algorithm,algcompatible}

\newtheorem{lemma}{Lemma}

\DeclareMathOperator{\sinc}{sinc}
\DeclareMathOperator{\SNR}{SNR}

\newcommand{\bA}{\mathbf{A}}
\newcommand{\bb}{\mathbf{b}}
\newcommand{\bB}{\mathbf{B}}

\newcommand{\be}{\mathbf{e}}

\newcommand{\bF}{\mathbf{F}}

\newcommand{\bI}{\mathbf{I}}

\newcommand{\bo}{\mathbf{o}}

\newcommand{\bw}{\mathbf{w}}

\newcommand{\bx}{\mathbf{x}}

\newcommand{\by}{\mathbf{y}}

\newcommand{\bz}{\mathbf{z}}

\DeclareMathOperator{\re}{\mathbb{R}e}

\begin{document}
\title{Interferometry-based modal analysis with finite aperture effects}

\author{Davood~Mardani, Ayman F. Abouraddy and George K. Atia
\thanks{
D. Mardani and G. Atia are with the Department of Electrical and Computer Engineering, University of Central Florida, e-mails: d.mardani@knights.ucf.edu, george.atia@ucf.edu.

A. Abouraddy is with the College of Optics \& Photonics (CREOL), University of Central Florida, email: raddy@creol.ucf.edu.}
}

\maketitle

\begin{abstract}
We analyze the effects of aperture finiteness on interferograms recorded to unveil the modal content of optical beams in arbitrary basis using generalized interferometry. We develop a scheme for modal reconstruction from interferometric measurements that accounts for the ensuing clipping effects. Clipping-cognizant reconstruction is shown to yield significant performance gains over traditional schemes that overlook such effects that do arise in practice. Our work can inspire further research on reconstruction schemes and algorithms that account for practical hardware limitations in a variety of contexts.  
\end{abstract}

\section{Introduction}
The ability to recover information about a source of light from measured field data underlies many fundamental problems in spectroscopy \cite{spec89,Zhu14OE}, holography \cite{Martinez-Leon:17,Clemente:13}, optical interferometry \cite{Abouraddy12OL,OCT:book}, optical imaging \cite{OCT_science, Milad15OL}, and optical communications \cite{terabite,Bozinovic1545}.

Despite noteworthy efforts to develop theory and algorithms for the inverse source problem,  
much of the existing work have assumed ideal acquisition systems. The effects of physical hardware limitations on performance, however, have been largely unexplored. For example, while much work was devoted to leveraging structural information inherent to light beams (e.g., sparsity~\cite{Denis09OL}, total variation~\cite{TV}, etc.) through the use of regularizers and studying its implications on data acquisition (e.g., recovering signals from a reduced number of measurements~\cite{Duarte_single}), very little is known about the interplay of hardware limitations and signal reconstruction. Moreover, establishing recoverability guarantees under practical hardware constraints of data acquisition systems remains elusive. 

Among such limitations are the limited number of degrees of freedom of actual acquisition systems \cite{Mardani15OE,MardaniOE18}, the finite aperture size of the optical components \cite{sci_rep2017}, and their finite spatial resolution (e.g., Spatial Light Modulators (SLMs), and Digital Micromirror Devices (DMDs))~\cite{Duarte_single,Milad15OL,Duran15OE,Tajahuerce14OE}, to name a few. To underscore the importance of both studying and addressing such limitations, in the sequel we provide examples from the literature concerning their consequential impact on signal reconstruction. 

The number of degrees of freedom of a given data acquisition/sensing system sets a limit on its information capacity. For example, in interferometry-based holography and optical imaging~\cite{OCTprinciple,Martinez-Leon:17,Clemente:13}, the swept delay in the reference arm of a two-path interferometer is the sole degree of freedom at hand. As a result, successful recovery typically necessitates a large sample complexity due to the limited informational content of the highly-correlated measurements.   
This motivated the use of additional hardware components such as introducing optical masks along the path of the optical field in optical imaging and spectroscopy \cite{Duarte_single,Milad15OL,Gong15SR,Tajahuerce14OE,Arce14SPM}, and quantum state tomography \cite{Howland14PRL,Howland16PR}. While the extra degrees of freedom afforded by the randomization pattern these masks map on the field can boost the acquisition system's capability, the design of such masks -- which are non-native to such systems -- is neither cost- nor overhead-free. Moreover, since such masks block a large portion of the light field through sampling, they tend to reduce the effective signal-to-noise ratio (SNR)~\cite{Dumas17OE,Marcos16OE}.

Another important limitation emerges from the finite spatial resolution of the optical components. For instance, the non-vanishing pixel size of the random masks used to either collect measurements \cite{Duarte_single} or illuminate an object \cite{Milad15OL} in imaging applications contributes to the spatial resolution of the formed images. The use of finer pixels to step up resolution comes at the expense of higher-dimensionality \cite{candes_RIP}, thereby trading-off spatial resolution for computational/design cost, as well as potential degradation in signal reconstruction following from the curse of dimensionality phenomenon~\cite{candes_RIP}.

The finite-aperture size of the optical components -- the primary focus of this paper -- introduces non-linearities into the measurement model due to the ensuing clipping in the spatial domain~\cite{sci_rep2017}. In particular, when the light field expands due to spatial diffraction upon propagation, it gets clipped given the finite aperture size of lenses, SLMs, masks, etc. \cite{sci_rep2017,Dumas17OE}, leading to undesired loss of information in the tail of the beam profile beyond the aperture size. In prior experimental work on optical modal analysis, we observed that spatial clipping due to aperture finiteness is one of the most degrading factors for information recovery~(See \cite{sci_rep2017}).       

This paper is primarily focused on analyzing the effects that the finite aperture size and the ensuing beam clipping have on the ability to perform optical modal analysis in generalized interferometry. We also leverage the results of the analysis to devise a class of clipping-cognizant algorithms that provably yield significant gains over schemes oblivious to such effects due to aperture finiteness. 

Generalized interferometry, first introduced in \cite{Abouraddy12OL,sci_rep2017}, greatly extended standard temporal interferometry using general phase operators that act as `delays' in arbitrary degrees of freedom. It was shown that a beam can be analyzed into its constituent (arbitrary) spatial modes by replacing the standard temporal delay in the reference arm of an interferometer with a suitable unitary transformation (termed generalized delay) for which the optical modes are eigenfunctions. Examples of such generalized delays are the Fractional Fourier Transform (frFT) and the Fractional Hankel Transform (frHT) for Hermite Gaussian (HG) and Laguerre Gaussian (LG) modes, respectively. 

The analysis provided in \cite{Abouraddy12OL}, however, assumed ideal implementations of fractional transforms, which is not practically possible. For example, in subsequent work \cite{sci_rep2017} the frFT is implemented using a cascade of three SLMs that map a quadratic phase on the propagated field. 
On account of the practical limitations of such devices, many of the aforementioned imperfections are introduced into the measurement setup, including beam clipping due to finite-aperture size, limited spatial resolution due to the non-vanishing pixel size of the SLMs, and phase granularity due to the finiteness of the quantization levels, the effects of which on the quality of interferometric measurements collected was observed and documented in \cite{sci_rep2017}. In this paper, we take a principled approach to analyzing the effect of clipping due to the finite aperture size of the SLMs on interferometric measurements and leverage the acquired knowledge to compensate for such effects in reconstruction. We first introduce a model for the finite aperture size using clipping Linear Canonical Transforms (LCTs). Then, we analyze the output field of a generalized delay system as a cascade of regular and clipping LCTs. Appendix B provides a full analysis for the output field of different combinations of LCTs. To the best of our knowledge, this is the first work to provide a rigorous analysis of the interplay of finite aperture size on signal reconstruction and to provide clipping-cognizant solutions thereof.    

It is important to note that our work is substantially different from, and should not be confused with, a large body of work on super-resolution techniques, in which one aims to recover \emph{missing information} about an object or light beam due to various practical restrictions (such as the optical diffraction limit \cite{Diflim:book} and the non-zero detector pixel size in optical imaging) by leveraging \emph{prior information about the input signal} \cite{Supres:book}. For example, in super-resolution techniques used for imaging, the non-redundant information of several images and frames are combined to improve the resolution of one image \cite{Supres_LR_HR}. In this paper, we do not seek to recover information missing due to finite-aperture size. Rather, we exploit a derived (through rigorous analysis) clipping-cognizant measurement model to ensure that information relevant to the modal content of a light beam (and intrinsic to the interferometric measurements) is not disregarded in the reconstruction phase as in traditional models that overlook finite-aperture effects.

We also remark that while our focus is on optical modal analysis using interferometry, the analysis and machinery developed herein can be quite useful in other contexts, therefore could inspire further research on reconstruction schemes and algorithms that account for important and practical hardware limitations.

\section{Interferometric modal analysis: Ideal setting}\label{sec:gen_interf}

In temporal interferometry, an input light beam $\psi(t)$, where $t$ is the time variable, is directed to two different optical paths (the arms of the interferometer) via an optical splitter such as a semi-transparent mirror. A delay $\tau$ is swept in one arm of the interferometer, the reference arm, and the interferogram is calculated as the energy of the superposition of the output beam $\psi(t;\tau):=\psi(t-\tau)$ and the incident light beam $\psi(t)$ as,
\begin{equation}\label{eq:interferogram_general}
I(\tau)=\langle|\psi(t;\tau)+\psi(t)|^2\rangle,
\end{equation}   
where $\langle.\rangle$ denotes integral over time.

The interferogram calculated in (\ref{eq:interferogram_general}) can be used to access the spectral content of the input light beam by considering the beam expansion in the orthonormal set of complex exponentials $\{e^{i\omega_n t}\}$ as 
$\psi(t)=\sum_{n=1}^{N}c_n e^{i\omega_n t}$, where $c_n, \,\,n=1,2,\ldots,N$, are the expansion coefficients in the basis $\{e^{i\omega_n t}\}$. Here, we represent the beam as a superposition of $N$ modes based on the fact that only a finite set of spectral components are of interest or even accessible. Besides, the input beam is assumed to have unit energy, i.e., $\sum_{n=1}^N |c_n|^2=1$. The output of the reference arm can then be written as $\psi(t;\tau)=\psi(t-\tau)=\sum_{n=1}^{N}c_n e^{-i\omega_n \tau} e^{i\omega_n t}$. Replacing in (\ref{eq:interferogram_general}), 
\begin{equation}\label{eq:time_linear}
I(\tau)\!=\!2+2\sum_{n=1}^{N}|c_{n}|^{2}\cos(n\tau).
\end{equation}
Therefore, the Fourier transform (FT) of $I(\tau)$ collected by sampling the time delay $\tau$ at Nyquist rate reveals the spectral power content, $|c_n|^2,\,\,n=1,2,\ldots,N$, of the input beam.

Unveiling the spectral content of a light beam from interferometric measurements has been recently extended to optical modal analysis in arbitrary degrees of freedom beyond time (such as position), in what has been termed `generalized interferometry' \cite{Abouraddy12OL,sci_rep2017}. The basis for such generalization is the observation that the delay line forming the reference arm of the temporal interferometer amounts to a linear time-invariant system for which the frequency harmonics $e^{i\omega_n t}$ are eigenfunctions with corresponding eigenvalues $e^{-i\omega_n \tau}, n=1,2,\ldots,N$. Hence, by replacing the standard delay with a `generalized delay' operator for which the input modes are eigenfunctions, an incident light beam can be analyzed into its Hilbert space basis elements \cite{Abouraddy12OL,sci_rep2017}.   
More formally, consider an input beam, $\psi(x)=\sum_{n=1}^{N}c_n\phi_n(x)$, in a Hilbert space spanned by a discrete orthonormal basis $\{\phi_n(x)\}$, with arbitrary degree of freedom $x$ (e.g., spatial, angular, temporal),
where $c_n,\,n=1,2,\ldots,N$, are the modal coefficients. Replacing the time delay in the reference arm with a generalized operator $h(x,x';\alpha):=\sum_{n=1}^{N}e^{-i n\alpha}\phi_{n}(x)\phi_{n}^{*}(x')$ for which $\{\phi_n(x)\}$ are eigenfunctions, the output beam will be, 
\begin{equation}
\psi(x;\alpha)=\sum_{n=1}^{N}c_n\,e^{-in\alpha}\phi_n(x),
\end{equation}
where $\alpha$ is a generalized delay parameter, and $e^{-in\alpha}$ the eigenvalue corresponding to $\phi_n(x)$. Combining the output of the reference arm and the input beam as in (\ref{eq:interferogram_general}), we record an interferogram $I(\alpha)$ (with the exact same form in (\ref{eq:time_linear}) with $\tau$ replaced with $\alpha$) whose FT gives the modal weights $|c_n|^2$. 

It can be shown that the generalized operator $h(x,x';\alpha)$ is generally a fractional transform in arbitrary domain $x$, which can be realized using common optical components \cite{sci_rep2017}. For example, the generalized delay required to analyze a light beam in a Hilbert space spanned by the HG modes $\{\phi_n(x)\}$ is a frFT of order $\alpha$, i.e., $h(x,x';\alpha)\!\propto\!~\exp\left\{\frac{i\pi}{2}(x^2\cot\alpha\!+\!x'^2\cot\alpha\!-\!2xx'\csc\alpha)\right\}$,
which can be implemented using two cylindrical lenses \cite{Abouraddy12OL}. 
Similarly, to analyze a light beam into its radial LG modes, we require a frHT of order $\alpha$, which can be realized using two spherical lenses.  

Alternatively, generalized delay operators corresponding to various fractional transforms can be implemented using programmable SLMs, thereby affording additional flexibility in data acquisition. For example, we have shown that sub-Nyquist sampling of the delay parameter $\alpha$ suffices for reconstructing the modal content of \emph{sparse} beams, yielding substantial savings in sampling and computational complexities \cite{Mardani15OE,MardaniOE18}. Figure \ref{fig:frFT_schematic} depicts an actual implementation of a frFT filter 
using a three-SLM configuration used to analyze an input beam into its HG modes. We refer the reader to \cite{sci_rep2017} for further details. Therefore, we have shown that modal analysis can be carried out in any Hilbert space with basis defined over an arbitrary degree of freedom (beyond temporal) without change to the underlying interferometer structure except for replacing the standard time delay with an appropriate unitary transform 
(the generalized delay operator)
for which the Hilbert space basis elements are eigenfunctions.

\begin{figure}[htb]
	\centering
	\includegraphics[scale=1]{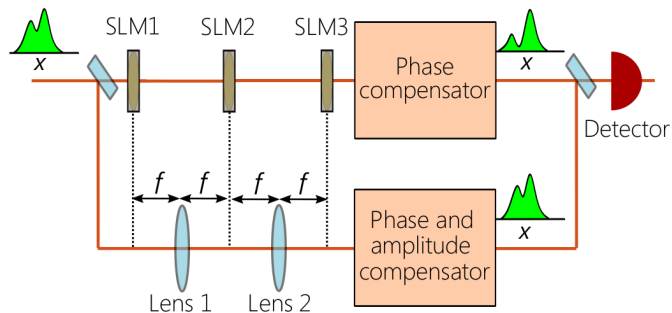}
	\vspace{-.05cm}
	\caption{Schematic of a frFT filter implemented using SLMs that act as quadratic phase operators.}
\label{fig:frFT_schematic}\vspace{-10pt}
\end{figure}

It is useful to rewrite (\ref{eq:interferogram_general}) in the linear system form
\begin{equation}\label{eq:linear_matrix}
\by=\bA\bx,
\end{equation}  
where $\by=[y(\alpha_1) \,y(\alpha_2)\,\ldots\,y(\alpha_M)]^{\mathrm{T}}$ is an $M\!\times\!1$ measurement vector with entries, $y(\alpha_m):=\frac{I(\alpha_m)-2}{2}$, for $M$ chosen settings $\alpha_m,\,m=1,2,...,M$, of the generalized delay parameter, $\bx=[|c_1|^2\,|c_2|^2\,\ldots\,|c_N|^2]^{\mathrm{T}}$ the $N\!\times\!1$ vector of modal weights, and $\bA$ an $M\!\times\!N$ matrix with entries $\cos(n\alpha_m),\,n=1,2,...,N, \,\,m=1,2...,M$, mapping the coefficient vector $\bx$ to an $\mathbb{R}^M$-dimensional measurement space. Our prior work exploited this alternative representation for the interferogram model to achieve compression gains in sample complexity and establish analytical performance guarantees for generalized modal analysis from compressive interferometric measurements sampled at sub-Nyquist rates \cite{Mardani15OE,MardaniarXiv17}.    

\section{Finite-aperture effect}\label{sec:imperfect}
\smallbreak

The previous section focused on optical modal analysis using generalized interferometry in an idealistic setting. In practice, however, the quality of the measurements collected will inevitably depend on the limitations of the hardware used and the underlying physical system constraints -- hence, the actual interferogram will deviate from the idealistic model in (\ref{eq:interferogram_general}) and (\ref{eq:linear_matrix}), which could adversely affect the performance of modal reconstruction. For example, in \cite{sci_rep2017} we have reported on the degradation in the quality of interferograms recorded experimentally originating from \emph{clipping effects} due to the finite-aperture size of the SLMs, the limited \emph{spatial phase resolution} along the transverse direction due to their non-vanishing pixel size, and the \emph{phase granularity} due to the finiteness of the number of phase quantization levels. 

Our experimental investigations have 
further
revealed that the clipping of the beams at the output of the SLMs beyond their aperture size limits has the most consequential effect on the quality of interferograms and, in turn, on modal reconstruction. For illustration, consider the example in Fig. \ref{fig:clip_pixel}, which shows the output interferogram of the generalized interferometer in Fig. \ref{fig:frFT_schematic} for an input beam consisting of the second Hermite Gaussian mode $\mathrm{HG}_2$, but this time taking the finite-aperture and non-vanishing pixel size of the SLMs into consideration. In theory, we expect the interferogram to exhibit a peak at $\alpha=\pi$. However, in both simulations and experiments, we see an apparent drop at $\alpha=\pi$ when the size of the SLMs is $16\mathrm{mm}$ and the pixel size is $10\mu\mathrm{m}$ (Fig. \ref{fig:clip_pixel}(a)). This is because the configuration shown no longer realizes the intended (ideal) frFT for which the HG mode is an eigenfunction. The observed drop is retained even if we use a finer pixel size of $5\mu\mathrm{m}$ as shown in Fig. \ref{fig:clip_pixel}(b). The peak, however, is extant if we increase the SLM size to $60\mathrm{mm}$ as per Fig. \ref{fig:clip_pixel}(c).

Motivated by that, this paper seeks to develop a thorough mathematical analysis of the impact of the finite-aperture size on the interferograms, and, leveraging the results of this analysis, propose a new paradigm for reconstruction that alleviates the ensuing degradation in mode recovery.

\begin{figure}[b]
	\centering
	\includegraphics[scale=1]{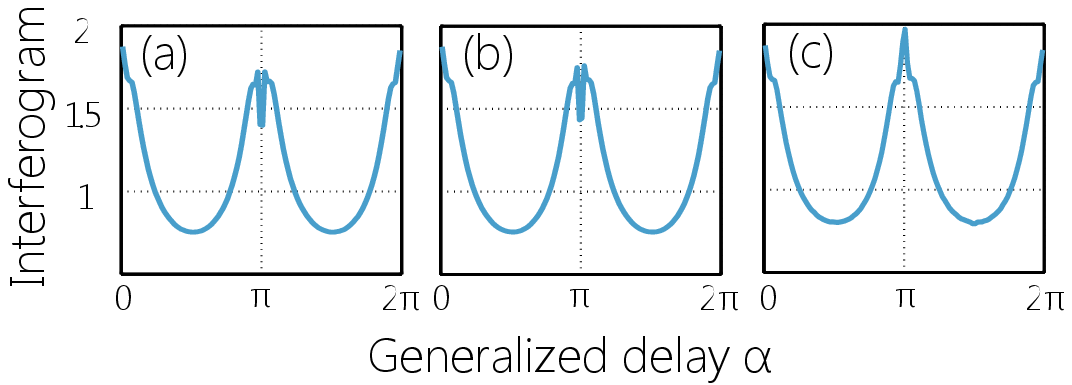}
    \vspace{-.3cm}
	\caption{The effect of spatial aperture and pixel size on the quality of the interferograms. (a) SLM size of $16 \mathrm{mm}$ and pixel size of $10 \mu\mathrm{m}$. (b) SLM size of $16 \mathrm{mm}$ and pixel size of $5 \mu\mathrm{m}$. (c) SLM size of $60 \mathrm{mm}$ and pixel size of $10 \mu\mathrm{m}$ \cite{sci_rep2017}.}
	\label{fig:clip_pixel}
    
\end{figure}

\begin{figure*}[t]
	\centering
	\includegraphics[scale=.8]{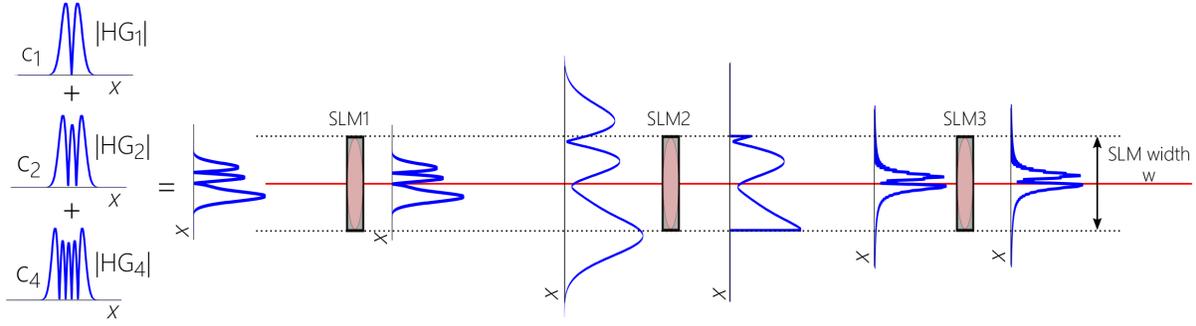}
	\vspace{-.4cm}
	\caption{Progression of a beam obtained as the superposition of $\mathrm{HG}_1, \mathrm{HG}_2, \mathrm{HG}_4$ modes as it propagates, diffracts and gets clipped by the SLMs of the frFT filter. The SLM width is $w=5$mm.}
	\label{fig:clipping SLMs}\vspace{-10pt}
\end{figure*}

\subsection{Clipping-cognizant measurement model}\label{ssec:meas_model}
Many of the components used to implement an optical setup can be modeled as spacial cases of a Linear Canonical Transform (LCT) characterized by four parameters defining the parameter matrix $M=\begin{pmatrix} a&b\\ c&d \end{pmatrix}$ with unit determinant, i.e., $ad-bc=1$, as,
\begin{equation}\label{eq:LCT}
\psi^M(u)=T^M\{\psi(x)\}(u)=\int_{-\infty}^{\infty}\psi(x)h^M(x,u)dx,
\end{equation}
where,
\begin{equation}\label{eq:kernel}
h^M(x,u)=\sqrt{\frac{1}{j2\pi b}}\exp(\frac{j}{2b}(ax^2-2xu+du^2))
\end{equation}
for $b\neq 0$, and
\begin{equation}\label{eq:scale}
\psi^M(u)=T^M\{\psi(x)\}(u)=\sqrt{d}\exp(j\frac{cdu^2}{2})\psi(du)
\end{equation}
for $b=0$ \cite{LCT1971, LCT:book, LCT:eigen}. The notation $T^M\{\psi(.)\}$ denotes the LCT operator with parameter matrix $M$ acting on input $\psi(.)$, and $u$ represents the degree of freedom in the LCT domain. Optical lenses, SLMs, and some of the commonly used linear operators such as Fourier transform, Fresnel integration and fractional Fourier transform are special cases of an LCT with different parameter matrices. For example, Fresnel integration used to approximate the short-range free-space diffraction in an optical setup is an LCT with $M=\begin{pmatrix} 1&\frac{\lambda l}{2\pi}\\ 0&1 \end{pmatrix}$, where $\lambda$ is the field wavelength and $l$ is the free-space length.     

To model an optical component with finite-aperture size $w$, we propose a \emph{clipping LCT}, $T^M_w\{.\}(u)$, whose output for an incident beam $\psi(x)$ is equal to that of an ideal LCT acting on $\psi(x)$ multiplied by a rectangular function of width $w$, i.e., $T^M_w\{\psi(x)\}(u)=T^M\{\psi(x)\Pi(\frac{x}{w})\}(u)$. Leveraging the product property of the LCT, which describes the LCT of a product of two functions (see  Appendix A), $T^M_w\{\psi(x)\}(u)$ gives,   
\begin{equation}\label{eq:clippingb}
\psi^M(u;w)=\frac{w}{2\pi|b|} e^{\frac{j d}{2b}u^2}\left((\psi^M(u) e^{\frac{-j d}{2b}u^2})*\sinc\left(\frac{wu}{2\pi b}\right)\right),
\end{equation}
where $*$ denotes convolution for $b\neq0$ and,
\begin{equation}\label{eq:clipping0}
\psi^M(u;w)=\sqrt{d}\exp\left(i\frac{cd}{2}u^2\right)\psi(du) \Pi\left(\frac{du}{w}\right),
\end{equation}
for $b=0$.

As mentioned earlier, the generalized delay operator is practically implemented using a cascade of optical components. For example, an frFT system analyzing the content of HG beams is realized using three SLMs separated by distances of $2f$ (see Fig. \ref{fig:frFT_schematic}), where $f$ is the focal distance of the lenses in the second arm of the interferometer. By the additivity property of LCTs \cite{LCT:book}, we can show that this system is equivalent to a cascade of five LCTs with parameter matrices $M_i=\begin{pmatrix} 1&0\\ c_i&1 \end{pmatrix},\,\,i=1,2,3$ (corresponding to an SLM with the phase $c_i$), and $M=\begin{pmatrix} 1&\frac{\lambda f}{\pi}\\ 0&1 \end{pmatrix}$ (modeling the Fresnel diffraction in free-space) as seen in the schematic of Fig.\ref{fig:frFT_schematic}. We model the SLMs using clipping LCTs given their finite-aperture size leading to the beam clipping illustrated in Fig.\ref{fig:clipping SLMs}. 

In general, the fractional transform in the reference arm of the interferometer in any degree of freedom can be modeled as a cascade of regular and clipping LCTs. However, it is important to note that the Hilbert space basis elements $\{\phi_n(x)\}$ are no longer eigenfunctions of this transformation owing to the present clipping effect.  
We obtain a closed-form expression for the output of any combination of clipping and regular LCTs (see Lemma 2 to Lemma 5 in Appendix B). Accordingly, the output of the generalized delay for an input basis element $\phi_n(x)$, is $\mathcal{L}\{\phi_n(x)\} = e^{-in\alpha}\hat{\phi}_n(x;\alpha,\bw),\,\,n=1,2,...,N$, where $\bw$ is a model parameter vector whose entries are the aperture sizes of the optical components (e.g., the widths $w_i, \,\,i=1,2,3$, of the three SLMs in the frFT realization). 
As an example, following from Lemma 4 in Appendix B, the response of the frFT system of order $\alpha$ implemented using finite-aperture SLMs to $\mathrm{HG}_n$, the $n^{\text{th}}$ mode $\phi_n(x)$, is
\begin{equation}\label{eq:clipoutput_1}
\begin{split}
&\hat{\phi}_n(x;\alpha,\bw)=\frac{w_1w_2|\csc\alpha|}{(\lambda l)^2}\Big[\Big(\exp\left(-j\frac{\pi\csc\alpha}{\lambda l}x^2\right)\\
&\times\Big(\phi_n(x)
\exp\left(-j\frac{\pi\cot\alpha}{\lambda l}x^2\right)*\sinc\left(\frac{w_1x\csc\alpha}{\lambda l}\right)\Big)\Big)\\
&*\sinc\left(\frac{w_2x}{\lambda l}\right)\Big]\times\Pi\left(\frac{x}{w_3}\right)\exp\left(j\frac{\pi(\csc\alpha +\cot\alpha)}{\lambda l}x^2\right).
\end{split}
\end{equation}

Accordingly, the output of the reference arm is,
\begin{equation}\label{eq:linearity}
\psi(x;\alpha,\bw)=\mathcal{L}\left\{\sum_{n=1}^N c_n\phi_n(x)\right\}=\sum_{n=1}^Nc_ne^{-in\alpha}\hat{\phi}_n(x;\alpha, \mathbf{w}).
\end{equation}

From (\ref{eq:interferogram_general}), the interferogram as function of $\alpha$ is  
\begin{eqnarray}\label{eq:finalmeasuements_imperfection1}
\begin{split}
&I(\alpha;\bw)=<|\psi(x)|^2>+<|\psi(x;\alpha,\bw)|^2>\\
&+<\psi(x)\psi^*(x;\alpha,\bw)>+<\psi(x;\alpha,\bw)\psi^*(x)>,
\end{split}
\end{eqnarray}
where the superscript $^*$ denotes the conjugate operator. The first term on the RHS of (\ref{eq:finalmeasuements_imperfection1}) is the input energy which is unity. The second term is the output energy of the reference arm, hereon denoted by $e_o(\alpha,\bw)$. From (\ref{eq:linearity}), the remaining terms on the RHS of (\ref{eq:finalmeasuements_imperfection1}) can be expanded as,
\begin{equation}\label{eq:remainingtems}
\begin{split}
&2\re\{\sum_{n=1}^N|c_n|^2e^{in\alpha}\int_{-\infty}^{+\infty}\hat{\phi}_n(x;\alpha,\bw)\phi_n^*(x)dx\}\\ 
&+\sum_{n=1}^N\sum_{\substack{{n'=1}\\{n'\neq n}}}^N c_n c_{n'}^*(e^{-in\alpha}\int_{-\infty}^{+\infty}\hat{\phi}_n(x;\alpha,\bw)\phi_{n'}^*(x)\,dx\\
&+e^{in'\alpha}\int_{-\infty}^{+\infty}\phi_n(x)\hat{\phi}^*_{n'}(x;\alpha,\bw)\,dx)\:,
\end{split}
\end{equation}
where $\re\{.\}$ denotes the real part. Defining $g_{nn'}(\alpha;\bw):= \int_{-\infty}^{+\infty}\hat{\phi}_n(x;\alpha,\bw)\phi_{n'}^*(x)\,dx$, the interferogram takes the form,
\begin{eqnarray}\label{eq:finalmeasuements_clipping}
\begin{split}
&I(\alpha;\bw)=1+e_o(\alpha,\bw)\\
&+2\sum_{n=1}^N|c_n|^2|g_{nn}(\alpha;\bw)|\cos(n\alpha+\angle g_{nn}(\alpha;\bw))\\
&+\sum_{n=1}^N\sum^N_{\substack{{n'=1}\\{n'\neq n}}}
c_n c^*_{n'}(e^{-in\alpha}g_{nn'}(\alpha;\bw)+e^{in'\alpha}g^*_{n'n}(\alpha;\bw)).\\
\end{split}
\end{eqnarray}
Defining the interferometric measurements $y(\alpha,\bw)\!:=\!\frac{1}{2}(I(\alpha,\bw)-1-e_o(\alpha,\bw))$, the measurement model can be written in matrix form as
\begin{equation}\label{eq:clip_model_noise_free}
\by=\bar{\bA}\bx+\bB\bar{\bx},
\end{equation}
where $\by\!\triangleq\![y(\alpha_1,\bw), y(\alpha_2,\bw),\ldots,y(\alpha_M,\bw)]^{\mathrm{T}}$, the $M\!\times\!N$ matrix $\bar{\bA}\!\triangleq\![|g_{nn}(\alpha;\bw)|\cos(n\alpha+\angle{g_{nn}(\alpha;\bw))}]$, the $M\!\times\!N(N-1)$ matrix $\mathbf{B}\!\triangleq\!\frac{1}{2}[g_{nm}(\alpha_i;\mathbf{w}_i)+g_{mn}^*(\alpha_i;\mathbf{w}_i)]$, and $\bar{\bx}\!\triangleq\![c_1d^*_2,c_1d^*_3,\ldots,\\c_1d^*_N,c_2d^*_1,c_2d^*_3,\ldots,c_Nd^*_{N-1}]^{\mathrm{T}}$ is an $N(N-1)\!\times\!1$ vector showing the interaction between the different modes. Since $\hat{\phi}_n(x;\alpha,\bw), \,n=1,2,...,N$, can be accurately calculated as in the frFT example of (\ref{eq:clipoutput_1}) from the lemmas derived in Appendix B, the sensing matrix $\bar{\bA}$, and the coefficient matrix $\bB$ in (\ref{eq:clip_model}) are entirely accessible for modal recovery. To account for noise potentially contaminating the measurements, we also incorporate an additive white Gaussian noise term $\bz$ whose entries have variance $\sigma^2$ to obtain the final measurement model
\begin{equation}\label{eq:clip_model}
\by=\bar{\bA}\bx+\bB\bar{\bx}+\bz\:.
\end{equation}
Next, we develop a class of algorithms that are shown to bring about performance gains in modal reconstruction in presence of finite aperture effects by leveraging the clipping-cognizant model derived in (\ref{eq:clip_model}).     

\subsection{Reconstruction methods}\label{ssec:recovery}
The previous analysis has revealed that the effect of aperture finiteness on the interferometric measurements is manifested in the sensing matrix $\bar{\bA}$, the coefficient matrix $\bB$, and the output energy $e_o(\alpha;\bw)$ of the reference arm. Therefore, a reconstruction method that takes advantage of prior information about these terms given the measurement model derived in (\ref{eq:clip_model}) should yield more reliable recovery. 

In an idealistic setting in which the measurement model is given by (\ref{eq:time_linear}), a FT of interferometric measurements acquired by sampling the generalized delay $\alpha$ at Nyquist rate suffices to retrieve the modal energies, i.e., $\hat{\bx}=|\bF\by|$,
where $\bF$ is the discrete Fourier transform matrix, and $\hat{\bx}$ contains the modal energies $|c_n|^2,n=1,2,\ldots,N$ of the input beam. 
Since in many modal analysis problems a large portion of the beam energy is carried by a small set of modes, i.e., the coefficient vector $\bx$ is sparse, we devise sparse recovery algorithms to retrieve the modal content of optical beams in presence of clipping under the linear model in (\ref{eq:clip_model}).  

Our first method ignores the third term on the RHS of (\ref{eq:finalmeasuements_clipping}). In this case, the interferometric measurements are  approximated by
\begin{equation}\label{eq:appmodel}
\by\!\approx\!\bar{\bA}\bx+\bz,
\end{equation}
where 
$\bar{\bA}$ is defined after (\ref{eq:clip_model}). Under this assumption, we can readily use a denoising recovery algorithm such as the Dantzig selector \cite{Dantzig} to recover the modal content, which solves
\begin{equation}\label{eq:Dz}
\begin{split}
&\text{minimize}~  \|\hat{\bx}\|_1\\
& \text{subject to}~ \|\bar{\bA}^{\mathrm{T}}(\bar{\bA}\hat{\bx}-\by)\|_{\infty}\!\leq\!\eta \, \sigma,
\end{split}
\end{equation}
where $\eta$ is a tuning parameter used to control the performance of reconstruction. We remark that although this method ignores terms derived in (\ref{eq:finalmeasuements_clipping}) pertinent to the present clipping, it still partially accounts for clipping captured in the definition of $\bar\bA$ in (\ref{eq:clip_model}) which is different from the ideal $\bA$ in (\ref{eq:linear_matrix}).

\begin{algorithm} [htb]
	\textbf{Input:}\\
$\by$,	$\bA$, $\bB$, $\sigma$ \\
	\textbf{Initialization:}\\
    $\gamma=0$,
    $\sigma' \leftarrow \sigma$\\
	$\hat{\bx}$\, $\leftarrow$ Solving Dantzig Selector with constraint:\\
 $\|\bar{\bA}^{\mathrm{T}}(\bar{\bA}\hat{\bx}-\by)\|_{\infty}\!\leq\!\eta \, \sigma'$ 
\vspace{5pt}\\    
\textbf{While $\ell < L$}\\
$\bar{\bx}$ $\leftarrow$ Estimating $\bar{\bx}$ from $\hat{\bx}$\\
$\gamma(\ell)=\sqrt{\frac{\sum_{i=1}^{N(N-1)} |\bar{x}_i|^2 \|\bb_i\|^2}{M}}$\,\,\,\,\,\,\,// Estimating an upper bound on the standard deviation\\
$\sigma'$ $\leftarrow$ $\sigma'+\gamma$\hspace{30mm}// Updating the constraint\\
$\hat{\bx}$ $\leftarrow$ Updating the estimate of $\bx$ (solution of Dantzig selector)\\
if \,\,\,$|\gamma(\ell)-\gamma(\ell-1)|\leq\zeta$,\hspace{5pt} stop iterations \\
	$\ell = \ell+1$\\
	\textbf{end While}\\
	\textbf{Output:}\\
	$\bx = \hat{\bx}$.
	\caption{Iterative reconstruction algorithm}
	\label {alg:iterative}
\end{algorithm}

\vspace{-.0cm}
\begin{figure}[b]
	\centering
	\includegraphics[scale=1]{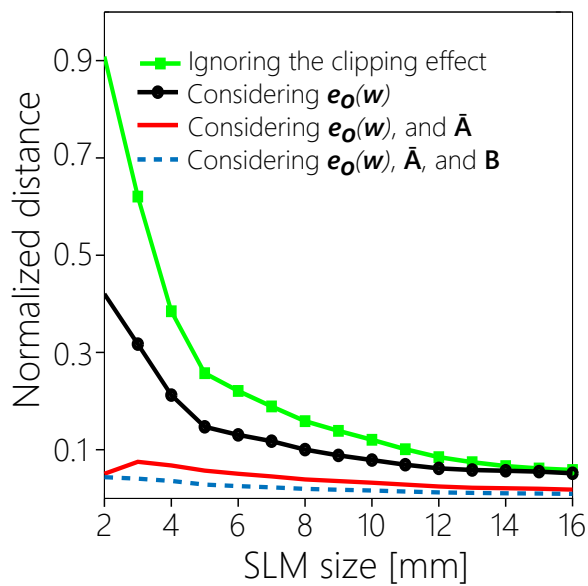}
    \vspace{-.3cm}
	\caption{Measurement model error in presence of clipping effect.}
	\label{fig:distance}
    \vspace{-.0cm}
\end{figure}

\begin{figure}[htb]
	\centering
	\includegraphics[scale=1]{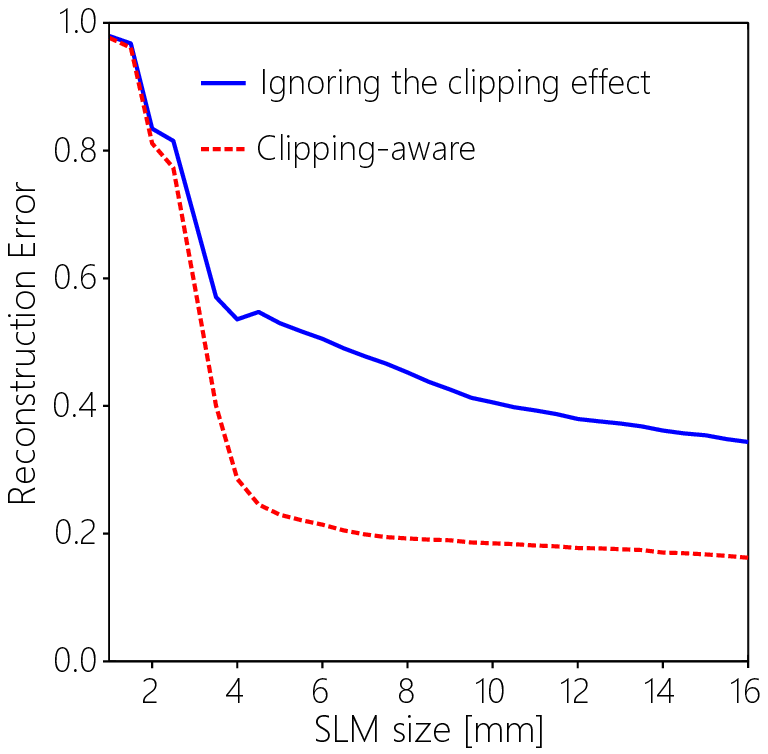}
    \vspace{-.3cm}
	\caption{FT-based modal recovery 
        and considering the clipping effect, reconstruction error versus SLM size.} 
	\label{fig:FT_results}
    \vspace{-.5cm}
\end{figure}

Nevertheless, seeking to further enhance the quality of reconstruction, our second method takes the effect of the term $\bB\bar{\bx}$ into consideration, hereon referred to as interference or noise factor. To this end, one possibility is to estimate $\bar{\bx}$, then subtract $\bB\bar{\bx}$ from the acquired measurements. However, this poses two main challenges. First, the vector $\bx$ is unknown. Second, the relation between $\bx$ and $\bar{\bx}$ is not one-to-one, which makes it impossible to accurately estimate $\bar{\bx}$ and exactly compute the term $\bB\bar{\bx}$ to eliminate it from the measurements even if $\bx$ is known. As such, we propose an iterative reconstruction algorithm detailed in Algorithm \ref{alg:iterative}, which uses the Dantzig selector as a core recovery procedure. 

Algorithm \ref{alg:iterative} is initiated with an estimate of $\bx$ obtained by solving  (\ref{eq:Dz}). Then, an approximate upper bound on the standard deviation of $\bB\bar{\bx}$ is calculated as, $\gamma=\sqrt{\frac{\sum_{i=1}^{N(N-1)} |\bar{x}_i|^2 \|\bb_i\|^2}{M}}$, where $|\bar{x}_i|$ is the $i^{\text{th}}$ element of an approximate $|\bar{\bx}|$, and $\|\bb_i\|$ the $\ell_2$-norm of the $i^{\text{th}}$ column of $\bB$. The magnitudes of the entries of $\bar{\bx}$ are obtained from the approximate vector $\bx$ calculated in the previous iteration. 
This upper bound is used to update the constraint in the Dantzig selector to $\eta\,(\sigma+\gamma)$ to improve reconstruction in the next iteration. The algorithm terminates when the difference between the $\gamma$'s in two consecutive iterations falls below a threshold $\zeta$, or when the number of iterations reaches a predefined maximum value $L$. 

\section{Numerical results}\label{sec:results}
To study the effect of finite-aperture on the interferometry-based modal analysis problem, we consider the example of analyzing a light beam into its HG modes. For data generation, the interferometer is implemented using three SLMs of the same aperture size $w$, and the free-space propagation between the SLMs is modeled using Fresnel integration (see Fig. \ref{fig:frFT_schematic}). The fidelity of this generative model has been confirmed by the agreement of the data with actual experimental measurements in \cite{sci_rep2017}. In this example, the potential number of modes $N = 64$, and the beam energy is carried by $s=4$ modes. To recover the modal energies in the inverse problem, we leverage the derived measurement model (\ref{eq:clip_model}). 
Here, we define $\SNR\triangleq10\log\left(\frac{\mathbb{E}[{I(\alpha;\bw)}^2]}{\sigma^2}\right)$, where $\mathbb{E}[.]$ stands for the expectation w.r.t. the distribution of the generalized delay $\alpha$ (here sampled from a uniform distribution $\mathcal{U}(0,2\pi)$) and $I(\alpha,\bw)$ the interferogram in (\ref{eq:finalmeasuements_clipping}). We evaluate the recovery error as $e\triangleq\frac{\|\mathbf{x}-\hat{\mathbf{x}}\|_2^2}{\|\mathbf{x}\|_2^2}$, where $\hat{\bx}$ is the reconstructed version of the sparse vector $\bx$.

\begin{figure*}[htb]
	\centering
	\includegraphics[scale=.95]{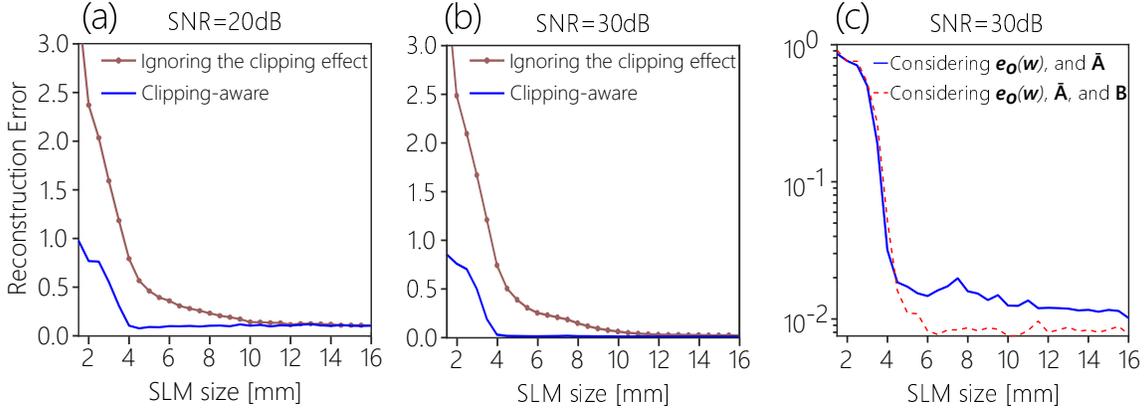}
    \vspace{-.3cm}
	\caption{Comparing reconstruction performance of the CS based approach with considering the clipping terms $\be_\bo(\bw)$, and $\bar{\bA}$ to that of the case in which the clipping effect is ignored. (a) SNR=20dB. (b) SNR=30dB. (c) Comparing the reconstruction error of the iterative algorithm to that of the regular CS based algorithm where the term $\bB\bar{\bx}$ is ignored, SNR=30dB. }
	\label{fig:CS_error_2terms}
    \vspace{-.4cm}
\end{figure*}

\begin{figure}[htb]
	\centering
	\includegraphics[scale=1]{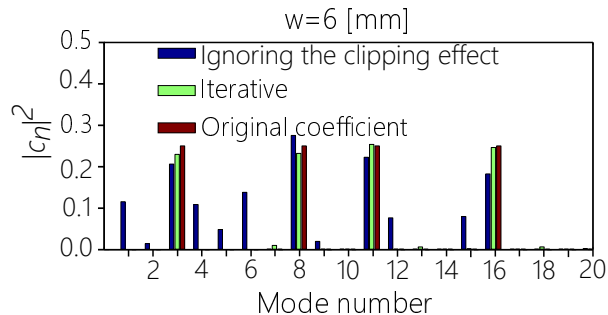}\vspace{-.2cm}
	\caption{Reconstruction of modal coefficients. $w=6$ mm, SNR=30dB.}
	\label{fig:CS_bar}\vspace{-.2cm}
\end{figure}

To validate the derived measurement model and further underscore the importance of accounting for the finite aperture effects, Fig. \ref{fig:distance} displays the normalized distance, $\|\bI-\hat{\bI}\|_2/\|\bI\|_2$, between the true measurements $\bI$ (obtained from the generative forward model) and the predicted measurements $\hat\bI$ for a range of increasingly refined measurement models. 
The green curve with square markers in Fig. \ref{fig:distance} shows the normalized distance for the idealistic model (\ref{eq:linear_matrix}), which completely ignores the clipping effect. The black curve with circle markers is for a model that only captures the output energy term $\be_\bo(\bw) =[e_o(\alpha_1,\bw),e_o(\alpha_2,\bw),\ldots,e_o(\alpha_M,\bw)]^{\mathrm{T}}$ in (\ref{eq:finalmeasuements_clipping}) in accounting for the clipping, but otherwise ignores all other terms. The solid (red) curve further considers the modified sensing matrix $\bar\bA$ as per (\ref{eq:appmodel}), thus yields smaller error. Finally, the dotted blue curve corresponds to the most comprehensive model in (\ref{eq:clip_model}), where all the clipping-related terms are accounted for, i.e., the modified matrix $\bar{\bA}$, the coefficient matrix $\bB$, and the output energy of the frFT system $\be_\bo(\bw)$.

Next, we investigate the reconstruction performance with clipping-cognizant recovery. Fig. \ref{fig:FT_results} shows the reconstruction error of the FT-based reconstruction 
versus the SLM size. Rather than the FT of $\frac{1}{2}\bI(\bw)-\mathbf{1}$, we reconstruct the sparse vector $\hat{\bx}=|\bF(\frac{1}{2}\{\bI(\bw)-\mathbf{1}-\be_\bo(\bw)\})|$, where $\bI(\bw)=[I(\alpha_1,\bw),I(\alpha_2,\bw),\ldots,I(\alpha_M,\bw)]^{\mathrm{T}}$ and $\mathbf{1}$ is a vector of all ones. Therefore, we provide a first level of compensating for the  clipping effect by accounting for the the output energy term $\be_\bo(\bw)$. 
As shown in Fig.\ref{fig:FT_results}, considering the clipping effect (red curve) reduces the reconstruction error and improves the quality of modal recovery. 

Figs. \ref{fig:CS_error_2terms} (a) and \ref{fig:CS_error_2terms} (b) show the reconstruction error of the proposed modal analysis approach versus the SLM size while adopting the CS-based recovery algorithms. As shown, accounting for the output energy of the frFT system and the modified sensing matrix $\bar{\bA}$ as per measurement model (\ref{eq:appmodel}) greatly improves the quality of recovery over the idealistic model in (\ref{eq:linear_matrix}), where the finite-aperture effect is ignored.   

To further improve the quality of reconstruction, we consider the more comprehensive measurement model in (\ref{eq:clip_model}), which also incorporates the derived $\bB\bar{\bx}$ and uses the proposed iterative recovery algorithm described in Algorithm \ref{alg:iterative}. Fig. \ref{fig:CS_error_2terms} (c) shows that Algorithm \ref{alg:iterative} yields further improvement in recovering the modal content of the incident light beam. Fig. \ref{fig:CS_bar} shows a significant improvement in the recovery of the modal energies using Algorithm \ref{alg:iterative} versus a Dantzig selector that ignores the clipping effect.

\section*{Appendix A: Product property of LCTs}
\begin{lemma}
	Let s(x) and g(x) be two signals or light beams. The LCT of their product, $\psi(x)=s(x)\cdot g(x)$, with real parameter matrix $M=\begin{pmatrix} a&b\\ c&d \end{pmatrix}$ is given by,
	\begin{eqnarray}\label{eq:productLCTb}
	\begin{split}
	&T^M\{\psi(x)\}(u)=\psi^M(u)\\
    &=\frac{1}{2\pi|b|}e^{i\frac{d}{2b}u^2}[\big(s^M(u)e^{-i\frac{d}{2b}u^2}\big)*g^{FT}(\frac{u}{2\pi b})]\\
	&=\frac{1}{2\pi|b|}e^{i\frac{d}{2b}u^2}[\big(g^M(u)e^{-i\frac{d}{2b}u^2}\big)*s^{FT}(\frac{u}{2\pi b})],
\end{split}
	\end{eqnarray}
	for $b\neq0$ and,
	\begin{eqnarray}\label{eq:productLCT0}
	T^M\{\psi(x)\}(u)=\sqrt{d}e^{i\frac{cd}{2}u^2}s(du)g(du)
	\end{eqnarray}
	for $b=0$. Here, $s^M(u)$ ($g^M(u)$) and $s^{FT}(u)$ ($g^{FT}(u)$) are the LCT and Fourier transform of $s(x)$ ($g(x)$), respectively.
\end{lemma}

\begin{proof}
	The proof of (\ref{eq:productLCT0}) follows directly from the definition of LCTs with $b=0$. To prove (\ref{eq:productLCTb}), we follow the same procedure used in \cite{FFT_product} to establish the product property of the fractional Fourier Transform. We start by the definition of LCTs as,
	\begin{eqnarray}\label{eq:LCTproductproof1}
	\begin{split}
	&T^M\{s(x)\cdot g(x)\}(u)=
	\int_{-\infty}^{\infty}(s(x)\cdot g(x))\\
    &\times\sqrt{\frac{1}{i2\pi b}}\exp\big(\frac{i}{2b}(ax^2-2xu+du^2)\big)dx.
	\end{split}
	\end{eqnarray}
Replacing $s(x)$ with the Inverse Linear Canonical Transform (ILCT) of its LCT, $s^M(\hat{u})$, with parameter matrix $\begin{pmatrix} d&-b\\ -c&a \end{pmatrix}$, we have,
	\begin{eqnarray}\label{eq:LCTproductproof2}
	\begin{split}
	&T^M\{\psi(x)\}(u)=\psi^M(u)\\
	&=\int_{-\infty}^{\infty} \Big(g(x)\times \sqrt{\frac{1}{i2\pi b}}\exp\big(\frac{i}{2b}(ax^2-2xu+du^2)\big)\\
	&\times\!\int_{-\infty}^{\infty}s^M(\hat{u})\sqrt{\frac{-1}{i2\pi b}}\exp\big(\frac{-i}{2b}(d\hat{u}^2-2x\hat{u}+ax^2)\big)d\hat{u}\Big)dx\\
	&=\frac{\exp(\frac{i d}{2b}u^2)}{2\pi|b|}\int_{-\infty}^{\infty}s^M(\hat{u})\exp(\frac{-i d}{2b}\hat{u}^2)\\
	&\times (\int_{-\infty}^{\infty}g(x)\exp(\frac{-i}{2\pi b}2\pi x(u-\hat{u}))dx)d\hat{u}.
	\end{split}
	\end{eqnarray}
Here, the degree of freedom in the LCT domain of $s(x)$ is denoted $\hat{u}$. As seen, the integral with respect to $x$ is actually the Fourier Transform of $g(x)$, where the variable in the Fourier domain is replaced by $\frac{(u-\hat{u})}{2\pi b}$. Therefore, the LCT is,
	\begin{eqnarray}\label{eq:LCTproductproof3}
	\begin{split}
	&\psi^M(u)=\frac{\exp(\frac{i d}{2b}u^2)}{2\pi|b|}\\
	&\times \int_{-\infty}^{\infty}s^M(\hat{u})\exp(\frac{-i d}{2b}\hat{u}^2)g^{FT}(\frac{u-\hat{u}}{2\pi b})d\hat{u}\\
	&=\frac{\exp(\frac{id}{2b}u^2)}{2\pi|b|}[(s^M(u)\exp(\frac{-i d}{2b}u^2))*g^{FT}(\frac{u}{2\pi b})].
	\end{split}
	\end{eqnarray}
The last equation in (\ref{eq:LCTproductproof3}) provides the LCT of the product of two signals in closed-form. Alternatively, this closed-form expression can be written as,  
	\begin{eqnarray}\label{eq:commutative}
	\psi^M(u)=\frac{e^{(\frac{i d}{2b}u^2)}}{2\pi |b|}[(g^M(u)e^{(\frac{-i d}{2b}u^2)})*s^{FT}(\frac{u}{2\pi b})],
	\end{eqnarray}
by switching the roles of $s(t)$ and $g(t)$ in (\ref{eq:LCTproductproof2}) and replacing $g(t)$ with the ILCT of $g^M(\hat{u})$. This shows the \emph{commutative property} for the LCT of a product.
\end{proof}
This property is used to define the output signal of a clipping LCT, where the clipping effect is modeled as the multiplication of the input beam with a rectangular function.  

\section*{Appendix B: Output of a cascade of clipping LCTs}
To analyze the output beam of different combinations of regular and clipping LCTs, we first establish the \emph{clipping additivity property}. Based on this property, the output beam of a system consisting of a regular LCT with parameter matrix $M_1$ and a clipping LCT of width $w$ with parameter matrix $M_2$ for the input beam $\psi(x)$ is equal to the output beam of a clipping LCT with width $w$ and parameter matrix $M_2M_1$ as,
\begin{equation}\label{eq:additivity rotation}
T^{M_2}\big\{T_w^{M_1}\{\psi(x)\}(u_1)\big\}(u_2)=T_w^{M_2M_1}\{\psi(x)\}(u_2),
\end{equation}
where $M_1=\begin{pmatrix} a_1&b_1\\ c_1&d_1 \end{pmatrix}$, $M_2=\begin{pmatrix} a_2&b_2\\ c_2&d_2 \end{pmatrix}$ and $b_1\neq0$, $b_2\neq0$. This property follows from the definitions of regular and clipping LCTs.

In the proposed basis analysis approach, the generalized phase operator system can be implemented using a cascade of optical components modeled as regular and clipping LCTs. 

Next, we establish several lemmas to capture the clipping effect at the output of systems implemented by clipping LCTs. First, we introduce some additional notation. For $\ell=1,2,\ldots,L$, we define $M_{L\ell}\!\triangleq\!M_LM_{L-1}\ldots M_{\ell}=\begin{pmatrix} a_{L\ell}&b_{L\ell}\\ c_{L\ell}&d_{L\ell}\end{pmatrix},$ and $M_{LL}\!\triangleq\!M_{L}$, where $M_{\ell}=\begin{pmatrix} a_{\ell}&b_{\ell}\\ c_{\ell}&d_{\ell}\end{pmatrix}$. 
We also define the recursive operator equations, 
\begin{equation}\label{recursive}
\begin{split}
&\kappa_{n}\big\{\psi(x);\{M_{L\ell}\}_{\ell},\{w_{\ell}\}_{\ell}\big\}(u) = \exp\left(i \Big(\frac{d_{Ln}}{2b_{Ln}}-\frac{d_{L(n+1)}}{2b_{L(n+1)}}\Big)u^2\right)\\
&\times\Big[\kappa_{n-1}\big\{\psi(x);\{M_{L\ell}\}_{\ell},\{w_{\ell}\}_{\ell}\big\}(u)*\sinc(\frac{w_nu}{2\pi b_{Ln}})\Big], \: n=2,3,\ldots,L 
\end{split}
\end{equation}
where,
\begin{equation}\label{eq:initial}
\begin{split}
&\kappa_1\big\{\psi(x);\{M_{L\ell}\}_{\ell},\{w_{\ell}\}_{\ell}\big\}(u) = \exp\left(i \Big(\frac{d_{L1}}{2b_{L1}}-\frac{d_{L2}}{2b_{L2}}\Big)u^2\right)\\
&\times\!\left(T^{M_{L1}}\{\psi(x)\}(u)
\exp\Big(-i\frac{d_{L1}}{2b_{L1}}u^2\Big)
*\sinc\Big(\frac{w_1u}{2\pi b_{L1}}\Big)\right),
\end{split}
\end{equation}
and $\{ . \}_{\ell}$ is a set indexed by $\ell=1,2,...,L$.

\begin{enumerate}
\item The next lemma calculates the output of a system formed by a cascade of $L$ clipping LCTs with $b_{\ell}\neq0, \,\,{\ell}=1,2,\ldots,L$. 

\begin{lemma}\label{the:successive_b_not0}
Let $\psi(x)$ be the input beam of $L$ clipping LCTs with $M_{\ell}=\begin{pmatrix} a_{\ell}&b_{\ell}\\ c_{\ell}&d_{\ell} \end{pmatrix},\,\,\,{\ell}=1,...,L$ and $b_{\ell}\neq0$. Then, the output of this cascade system is given by,\\
\vspace{-10 pt}
\begin{equation}\label{eq:successive_b_not0}
\begin{split}
\psi_o(u)&=\left(\prod_{\ell=1}^{L}\frac{w_{\ell}}{2\pi|b_{L\ell}|}\right)\\
&\!\times\!\kappa_L\big\{\psi(x);\{M_{L\ell}\}_{\ell},\{w_{\ell}\}_{\ell}\big\}(u),\:\ell=1,2,\ldots,L,
\end{split}
\end{equation}
where $T^{M_{L1}}\{.\}(u)$ is the linear system equivalent to the cascade of $L$ regular LCTs and $\kappa_L$ is defined through the recursion in (\ref{recursive}) and (\ref{eq:initial}).

\end{lemma}
\begin{proof}
The proof of Lemma \ref{the:successive_b_not0} follows directly from the definition of clipping LCTs and the clipping additivity property in (\ref{eq:additivity rotation}).
\end{proof}

\item The following lemma calculates the output of a  sequence of $L$ chirp multiplications and scaling systems -- equivalently LCTs with zero $b$ parameter -- with finite-aperture size (e.g., useful in modeling lenses). 
\begin{lemma}\label{the:successive_b0}
Let $\psi(x)$ be the signal or light beam input to a cascade of $L$ clipping LCTs with parameter matrices $M_\ell=\begin{pmatrix} a_{\ell}&0\\ c_{\ell}&d_{\ell} \end{pmatrix}, \,\, {\ell}=1,\ldots,L$. Then, the output signal $\psi_o(u)$ is given by,
\begin{eqnarray}\label{eq:successive_b0}
\begin{split}
\psi_o(u)&=T^{M_LM_{L-1}...M_2M_1}\{\psi(x)\}(u)\\
&\times \Pi\left(\frac{u}{\min\{\frac{w_1}{|d_{L1}|},\frac{w_2}{|d_{L2}|},...,\frac{w_L}{|d_L|}\}}\right).\\
\end{split}
\end{eqnarray}
\end{lemma}

\begin{proof}
From the definition of the clipping LCT with $b_{\ell}=0$, for ${\ell}=1,2,\ldots,L$, we have,  
	\begin{equation}\label{eq:proofb0_1}
	\begin{split}
	\psi_o(u)&=T^{M_L}_{w_L}\bigg\{T^{M_{L-1}}_{w_{L-1}}\Big\{\ldots\big\{T^{M_1}_{w_1}\{\psi(x)\}\big\}\Big\}\bigg\}(u)\\
    &=\sqrt{d_Ld_{L-1}\ldots d_3d_2d_1}\exp\big(i \frac{d_Ld_{L-1}\ldots d_3d_2d_1}{2}u^2\\
	&\times(\frac{c_L}{d_{L-1}d_{L-2}\ldots d_3d_2d_1}+\frac{c_{L-1}d_L}{d_{L-2}\ldots d_3d_2d_1}+\ldots\\
    &+\frac{c_2d_Ld_{L-1}\ldots d_3}{d_1}+c_1d_Ld_{L-1}\ldots d_3d_2)\big)\\
	&\times \psi(d_Ld_{L-1} \ldots d_3d_2d_1u)\\
    &\times \Pi\left(\frac{u}{\min\{\frac{w_1}{|d_Ld_{L-1}\ldots  d_3d_2d_1|},\frac{w_2}{|d_Ld_{L-1}\ldots d_3d_2|},\ldots,\frac{w_L}{|d_L|}\}}\right)\:.
	\end{split}
	\end{equation}
Accordingly, the result of Lemma \ref{the:successive_b0} follows by observing that the terms on the RHS of (\ref{eq:proofb0_1}) multiplying the rectangular function are the output of a regular LCT with parameter matrix,
	\begin{equation}\label{eq:proofb0_2}
	\begin{split}
	&M_LM_{L-1}...M_{1}=\\
	&\begin{bmatrix} \frac{1}{d_Ld_{L-1}\ldots d_2d_1}&0\\ \frac{c_L}{d_{L-1}\ldots d_1}+\frac{c_{L-1}d_L}{d_{L-2}\ldots d_1}+\ldots+c_1d_{L}d_{L-1}\ldots d_2&d_L\ldots d_2d_1\end{bmatrix}\\
    &\triangleq\begin{bmatrix} a_{L1}&b_{L1}\\ c_{L1}&d_{L1}\end{bmatrix},
	\end{split}
	\end{equation}
	where $d_{L\ell}=d_Ld_{L-1}\ldots d_{\ell}$ for $\ell\neq L$. 
\end{proof}

\item 
Similar to the previous cases, the following lemma computes the output a sequence of clipping LCTs with $M_{\ell}=\begin{pmatrix} a_{\ell}&b_{\ell}\\ c_{\ell}&d_{\ell} \end{pmatrix},\,\,\,{\ell}=1,\ldots,L$, however, in this case the system is formed by interleaving both types of clipping LCTs, such that the LCTs with odd and even orders have zero and non-zero parameter $b$, respectively. 
We remark that this case is commonly encountered in various applications. For example, a sequence of lenses in an optical setup 
act as chirp multiplications (equivalent to LCTs with zero $b$ parameter), while the free-space propagation between the lenses can be modeled as Fresnel diffractions (LCTs with non-zero $b$ parameter). 

\begin{lemma}\label{the:successive_b0_b_not_0}
Let $\psi(x)$ be the input to a sequence of $L$ clipping LCTs with parameter matrices $M_{\ell}=\begin{pmatrix} a_{\ell}&b_{\ell}\\ c_{\ell}&d_{\ell} \end{pmatrix},\,\,\,{\ell}=1,\ldots,L$, where $b_{2{\ell}'}\neq0$ and $b_{2{\ell}'-1}=0$,$\,\, {\ell}'=1,2,\ldots,\frac{L}{2}$, and $L$ an even integer. Then, the output is given by,
\begin{eqnarray}\label{eq:successive_b0_not0}
\begin{split}
\psi_o(u)&=\left(\prod_{{\ell'}=1}^{\frac{L}{2}}\frac{w'_{{\ell'}}}{2\pi|b'_{\frac{L}{2}{\ell'}}|}\right)\\
&\!\times\!\kappa_{\frac{L}{2}}\big\{\psi(x);\{M'_{\frac{L}{2}\ell'}\}_{\ell'},\{w'_{\ell'}\}_{\ell'}\big\}(u),
\: \ell'=1,2,\ldots,\frac{L}{2}\:,
\end{split}
\end{eqnarray}
where 
$w'_{\ell'}=\min\left\{w_{2{\ell'}},\frac{w_{2{\ell'}-1}}{|d_{2{\ell'}-1}|}\right\}$, $M'_{\frac{L}{2}{\ell'}}\triangleq  M'_{\frac{L}{2}}M'_{\frac{L}{2}-1}\ldots M'_{{\ell'}} =\begin{pmatrix} a'_{\frac{L}{2}{\ell'}}& b'_{\frac{L}{2}{\ell'}}\\ c'_{\frac{L}{2}{\ell'}}& d'_{\frac{L}{2}{\ell'}}  \end{pmatrix}$, and $M'_{\ell'}=\begin{pmatrix} a'_{\ell'} & b'_{\ell'} \\ c'_{\ell'} & d'_{\ell'} \end{pmatrix}\!\triangleq\!M_{2{\ell'}}M_{2{\ell'}-1}$, ${\ell'}=1,2,\ldots,\frac{L}{2}$, and $\kappa_{\frac{L}{2}}$ is defined through the recursion in (\ref{recursive}) and (\ref{eq:initial}).
\end{lemma}

\begin{proof}
	To prove (\ref{eq:successive_b0_not0}), the cascade of clipping LCTs is viewed as a sequence of two-LCT blocks with parameter matrices $M_{2{\ell'}-1}$ and $M_{2{\ell'}}$, wherein $b_{2{\ell'}-1}=0$ and $b_{2{\ell'}}\neq0$. Accordingly, the output of each block for arbitrary input $\psi(x)$ is,
	\begin{equation}\label{eq:proofb01_1}
	\begin{split}
	&T_{w_{2{\ell'}}}^{M_{2{\ell'}}}\big\{T_{w_{2{\ell'}-1}}^{M_{2{\ell'}-1}}\{\psi(x)\}(v)\big\}(u)=
	\int\sqrt{d_{2{\ell'}-1}} \,\,e^{i\frac{c_{2{\ell'}-1}d_{2{\ell'}-1}}{2}v^2}\\
    &\times\psi(d_{2{\ell'}-1}v)\Pi\Big(\frac{d_{2{\ell'}-1}v}{w_{2{\ell'}-1}}\Big)\times h^{M_{2{\ell'}}}(v,u)\Pi\Big(\frac{v}{w_{2{\ell'}}}\Big)\,dv\\
    &=\int T^{M_{2{\ell'}-1}}\{\psi(x)\}(v)h^{M_{2{\ell'}}}(v,u)\Pi\Big(\frac{v}{\min\{w_{2{\ell'}},\frac{w_{2{\ell'}-1}}{|d_{2{\ell'}-1}|}\}}\Big)dv\\
	&=T^{M_{2{\ell'}}M_{2{\ell'}-1}}_{\min\{w_{2{\ell'}},\frac{w_{2{\ell'}-1}}{|d_{2{\ell'}-1}|}\}}\{\psi(x)\}(u),
	\end{split}
	\end{equation}
    where $h^{M_{2{\ell'}}}(v,u)$ is the kernel of a regular LCT with the parameter matrix $M_{2{\ell'}}$. 
Since the overall $b$ parameter of each two-LCT block is non-zero, the result of Lemma \ref{the:successive_b_not0} is invoked to obtain the output of the cascade system.
\end{proof}

\item The final combination is similar to the previous case as a cascade of both types of LCTs but with the difference that the LCTs with the non-zero parameter $b$ are placed in the odd orders.

\begin{lemma}\label{the:successive_b_not_0_b_0}
Consider a similar setup as in the statement of Lemma \ref{the:successive_b0_b_not_0}, but with $b_{2{\ell'}-1}\neq0$ and $b_{2{\ell'}}=0, {\ell'}=1,2,\ldots,\frac{L}{2}$. The output of the cascade system is given by,
\begin{equation}\label{eq:successive_b_not0_b_0}
\begin{split}
\psi_o(u)=\left(\prod_{{\ell'}=1}^{\frac{L}{2}}\frac{w'_{{\ell'}}}{2\pi|b'_{\frac{L}{2}{\ell'}}|}\right)
&\times \kappa_{\frac{L}{2}}\big\{\psi(x);\{M'_{\frac{L}{2}\ell'}\}_{\ell'},\{w'_{\ell'}\}_{\ell'}\big\}(u)\\
&\times\Pi\Big(\frac{d_Lu}{w_L}\Big), \ell'=1,2,\ldots,\frac{L}{2}\:,
\end{split}
\end{equation}
where $w'_1=w_1$, $w'_{\ell'}=\min\{w_{2{\ell'}-1},\frac{w_{2{\ell'}-2}}{|d_{2{\ell'}-2}|}\}$ for ${\ell'}=2,3,\ldots,\frac{L}{2}$ and $\kappa_{\frac{L}{2}}$ is defined through the recursion in (\ref{recursive}) and (\ref{eq:initial}).
\end{lemma}
\end{enumerate}

\begin{proof}
Again, we view the cascade system as a sequence of two-LCT blocks with $b_{2{\ell'}-1}\neq 0$ and $b_{2{\ell'}}=0$ for ${\ell'}=1,2,\ldots,\frac{L}{2}$. Therefore,
\begin{equation}\label{eq:proofb10_1}
	\begin{split}
	&T^{M_{2{\ell'}}}_{w_{2{\ell'}}}\big\{T^{M_{2{\ell'}-1}}_{w_{2{\ell'}-1}}\{\psi(x)\}\big\}(u)=\sqrt{d_{2{\ell'}}}\exp\Big(i\frac{c_{2{\ell'}}d_{2{\ell'}}}{{2}}u^2\Big)\\
    &\times\int \psi(x)h^{M_{2{\ell'}-1}}(x,d_{2{\ell'}}u)\Pi\Big(\frac{x}{w_{2{\ell'}-1}}\Big)\,dx
	\times\Pi\Big(\frac{d_{2{\ell'}}u}{w_{2{\ell'}}}\Big)\\
    &=T_{w_{2{\ell'}-1}}^{M_{2{\ell'}}M_{2{\ell'}-1}}\{\psi(x)\}(u)\Pi\Big(\frac{d_{2{\ell'}}u}{w_{2{\ell'}}}\Big), \,\,\,{\ell'}=1,2,\ldots,\frac{L}{2}.
	\end{split}
	\end{equation}
	This represents the output of a clipping LCT with a non-zero $b$ parameter multiplied by a rectangular function of width $\frac{w_{2{\ell'}}}{|d_{2{\ell'}}|}$. Since the $b$ parameter of $M_{2{\ell'}}M_{2{\ell'}-1}$ is non-zero, we invoke Lemma \ref{the:successive_b_not0} to obtain the output of the entire system as, 
	\begin{equation}\label{eq:proofb10_3}
	\begin{split}
	\psi_o(u)&=T_{\min\{w_{L-1},\frac{w_{L-2}}{|d_{L-2}|}\}}^{M_LM_{L-1}}\bigg\{T_{\min\{w_{L-3},\frac{w_{L-4}}{|d_{L-4}|}\}}^{M_{L-2}M_{L-3}}\Big\{...\big\{T^{M_2M_1}_{w_1}\{\psi(x)\}\big\}\Big\}\bigg\}(u)\\
    &\times\Pi\Big(\frac{d_Lu}{w_L}\Big).
	\end{split}
	\end{equation}
	Similar to the previous case, the expression in (\ref{eq:proofb10_3}) amounts to the output of a sequence of $\frac{L}{2}$ LCTs multiplying the rectangular function $\Pi(\frac{d_Lu}{w_L})$. 
\end{proof}

\section*{Funding}
This work was supported in part by ONR contracts N00014-14-1-0260 and N00014-17-1-2458 and NSF CAREER Award CCF-1552497.

\bibliographystyle{IEEEbib}
\bibliography{references.bib}

\end{document}